\documentclass[a4paper,UKenglish,cleveref, autoref, thm-restate]{lipics-v2021}
\hideLIPIcs
\usepackage{amsmath}
\usepackage{amsthm}
\usepackage{hyperref}
\usepackage{xcolor}
\usepackage{graphicx}
\usepackage{enumitem}
\usepackage{booktabs}
\usepackage[export]{adjustbox}

\usepackage{tikz}
\usetikzlibrary{backgrounds,scopes} 
\usetikzlibrary{arrows.meta} 

\def\problembox#1{\vspace{2mm}\noindent\fbox{\begin{minipage}{.985\textwidth}#1
    \end{minipage}}\vspace{2mm}}

\newcommand{\STree}{\mathsf{ST}}
\newcommand{\numocc}{\#\mathit{occ}}
\newcommand{\MUS}{\mathsf{MUS}}
\newcommand{\LRS}{\mathsf{LRS}}
\newcommand{\LSuf}[1]{\ensuremath{\LRS{}[#1]}}
\newcommand{\SSuf}[1]{\ensuremath{\mathsf{SDS}[#1]}}
\newcommand{\myProof}{\textbf{Proof.~}}
\newcommand{\LPF}{\mathsf{LPF}}

\newcommand{\fnDepth}{\ensuremath{\mathsf{depth}}}
\newcommand*{\pali}[1]{\overleftarrow{#1}} \newcommand*{\Weiner}{\mathsf{W}}
\newcommand*{\strlabel}{\mathsf{label}}
\newcommand*{\timeSU}{\ensuremath{t_{\textup{SU}}}}
\newcommand*{\timeTraverse}{t_{\textup{trav}}}
\DeclareMathOperator{\polyloglog}{polyloglog}
\DeclareMathOperator{\polylog}{polylog}

\author{Dominik {Köppl}}{University of Yamanashi, Kofu, Japan \and \url{https://dkppl.de}}{dkppl@dkppl.de}{0000-0002-8721-4444}{This work was supported in part by JSPS KAKENHI Grant Numbers 25K21150, 23H04378.}
\author{Gregory Kucherov}{LIGM, CNRS and Gustave Eiffel University, Marne-la-Vall\'ee, France \and \url{http://igm.univ-mlv.fr/~koutcher/}}{gregory.kucherov@univ-eiffel.fr}{0000-0001-5899-5424}{}

\authorrunning{Köppl and Kucherov}
\Copyright{D. Köppl and G. Kucherov}

\title{Near-real-time Solutions for Online String Problems}
\titlerunning{Near-real-time Solutions for Online String Problems}


\begin{document}
\nolinenumbers 
	
	\maketitle
	\begin{abstract}
	  Based on the Breslauer--Italiano online suffix tree construction algorithm (2013) with double logarithmic worst-case guarantees on the update time per letter, 
	  we develop near-real-time algorithms for several classical problems on strings,
	  including the computation of the longest repeating suffix array, the (reversed) Lempel--Ziv 77 factorization, and the maintenance of minimal unique substrings, all in an online manner.
	  Our solutions improve over the best known running times for these problems in terms of the worst-case time per letter, for which we achieve a poly-\emph{log-logarithmic} time complexity, within a linear space.
	  Best known results for these problems require a poly-\emph{logarithmic} time complexity per letter or only provide amortized complexity bounds.
As a result of independent interest, we give conversions between the longest previous factor array and the longest repeating suffix array in space and time bounds based on their irreducible representations, which can have sizes sublinear in the length of the input string.
	\end{abstract}


\setcounter{page}{1}

	\section{Introduction}
	Strings are a fundamental data type in computer science, and many classical problems on strings have been studied extensively in the literature.
	When processing a large string, it is often desirable to obtain results as soon as possible rather than waiting until the entire string has been read.
	For that, we consider algorithms that work in an \emph{online} manner. These process the string letter by letter from left to right, while maintaining information about the string read so far.
	Unfortunately, most online algorithms for strings provide only \emph{amortized} time guarantees per letter, i.e., the average time spent to process each letter is bounded by a function of the input size, but some read letters may take much longer to process.
	This can be problematic in real-time applications, where it is important to have predictable processing times for each letter.
	In this paper, we focus on online algorithms with \emph{worst-case} time guarantees per letter.
	While literature gives real-time algorithms for LZ78-like factorizations ~\cite{ziv78lz} with a static dictionary~\cite{hartman85optimal} for constant alphabets,
	palindrome recognition and pattern matching~\cite{galil76real}, 
	for many classical problems on strings, online algorithms with efficient worst-case time guarantees per letter are not known.

	In fact, a majority of string problems can be only solved efficiently if we maintain a full-text indexing data structure of the string read so far.
	One of the most fundamental such data structures is the suffix tree~\cite{weiner73linear}.
	The first online construction of the suffix tree is due to Ukkonen~\cite{ukkonen95online}, which builds a suffix tree online in time $O(n)$ for a constant-sized alphabet. 
	However Ukkonen's algorithm does not provide any worst-case time guarantee on processing an individual letter, other than a trivial $O(n)$ bound. 
	For certain inputs, this bound is actually tight:
Given a string $SSc$ of length $n$, Ukkonen's algorithm maintains an active node of string depth at least $|S|$ after having parsed the prefix $SS$, but then needs to add branches for all suffixes of $S$ in case $c$ is a new letter, 
and thus creates $\Theta(n)$ nodes just for updating the suffix tree for the last letter. (Spending $\Theta(n)$ for the last letter may be a common phenomenon in practice since we usually append a unique delimiter letter $\$$ at the end of the input string.)
To achieve worst-case time guarantees, a line of research is based on the idea to maintain a suffix tree (or similar data structure) for the \emph{inverted} input string.
The rationale for this is that appending a new letter to the input string amounts to introducing only one new suffix in the inverted string, which causes less updates to the data structure compared to Ukkonen's algorithm.
The starting point of this line of research is Weiner's algorithm~\cite{weiner73linear}, historically the first linear-time algorithm for suffix tree construction for constant-sized alphabets. 
Weiner's algorithm processes the string right-to-left by introducing a new suffix at each step, and is therefore suitable for this task.
In its original version, Weiner's algorithm does not provide a worst-case time guarantee on processing a single letter. 
However, in the following Section~\ref{sec:weiner}, we summarize modifications that accomplish this.
These modifications result in double logarithmic worst-case times per letter, a time complexity referred to as \emph{near-real-time}~\cite{breslauer13near}.

\subparagraph*{Our contributions}
Having a near-real-time suffix tree construction algorithm as a subroutine, we can solve various classical problems on strings in near-real-time online manner by relying on the maintained suffix tree.
Figure~\ref{fig:overview} gives an overview of the problems and applications studied in this paper, along with their dependencies.
We start with the computation of the \emph{longest repeating suffix array} (LRS) in Section~\ref{sec:lrs} and
draw a connection to the \emph{longest previous factor array} (LPF) in Section~\ref{sec:lpf},
which serves as the base for computing the Lempel--Ziv 77 (LZ77)~\cite{ziv77lz} factorization in Section~\ref{sec:lz77}.
In Section~\ref{sec:mus}, we maintain the set of \emph{minimal unique substrings} (MUS) of the string read so far.
Finally, in Section~\ref{sec:reversedlz}, we maintain two variants of the LZ77 factorization on the reversed string.
Table~\ref{tab:contribution} summarizes our results along with the best known prior results.
While some of these obtain $O(\polylog(n))$ time per letter (although with additional constraints on the space), 
we obtain the first online algorithms with worst-case $O(\polyloglog(n))$ time per letter for all the studied problems.

\begin{figure}[t]
	\centering
	  \resizebox{\linewidth}{!}{\begin{tikzpicture}[node distance=2cm,
			>=Stealth,
			every node/.style={draw, fill=white, rounded corners, align=center, minimum width=2.2cm, minimum height=1cm}
			]
			
\node (A1) at (4,0) {\textsc{SuffixUpdate}\\Section~\ref{sec:weiner}};
			\node [left of=A1,anchor=east] (A2) {\textsc{InsertionPoint}\\Section~\ref{sec:weiner}};
			\node [left of=A2,anchor=east] (A4) {\textsc{ShortestDoubleSuffix}\\Section~\ref{sec:mus}};
			\node [left of=A4,anchor=east,xshift=-3em] (A3) {\textsc{LongestRepeatingPrefix}\\Section~\ref{sec:lrs}};
			
\node [below of=A3,anchor=east] (B1) {$\LRS$ Array\\Section~\ref{sec:lrs}};
			\node [right of=B1,anchor=west] (B1A) {$\LPF$ Array\\Section~\ref{sec:lpf}};
			\node [right of=B1A,anchor=west] (B2) {LZ77\\Section~\ref{sec:lz77}};
			\node [right of=B2,anchor=west] (B3) {MUS\\Section~\ref{sec:mus}};
			\node [right of=B3,anchor=west,xshift=-2em] (B4) {rev.\ LZ, overlap\\Section~\ref{sec:reversedlz}};
			\node [right of=B4,anchor=west,xshift=0em] (B5) {rev.\ LZ, nonoverlap\\Section~\ref{sec:reversedlz}};
			
			\begin{scope}[on background layer]
\draw[->] (A2) -- (A1);
				\draw[->] (A3) -- (B1);
				\draw[->] (B1) to[bend right=15] (B2);
				\draw[->] (B1) -- (B1A);
				\draw[->] (A3) -- (B3);
				\draw[->] (A2) -- (A4);
				\draw[->] (A2) to[bend right=15] (A3);
				\draw[->] (A4) -- (B3);
				\draw[->] (A1) -- (B4);
				\draw[->] (A1) -- (B5);
			\end{scope}
			
		\end{tikzpicture}
	}
	\caption{Problems (above) and applications (below) studied in this paper, with their dependencies visualized by arrows. \textsc{SuffixUpdate} is the fundamental problem on which all other problems and applications rely, for which arrows are omitted.
	}
	\label{fig:overview}
\end{figure}
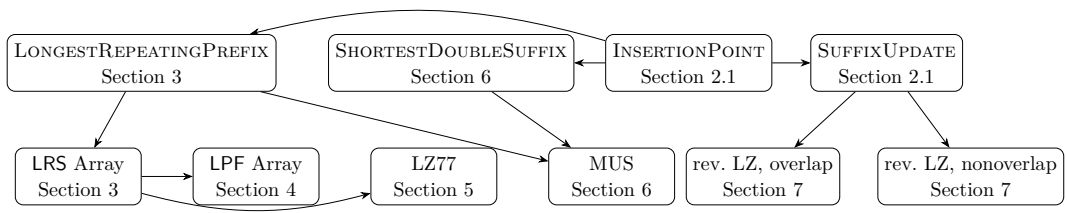

\begin{table}[h]
	\centering
	\caption{Worst-case time complexities per letter for online algorithms computing various string problems on a string of length~$n$.
		The shown results are the best known prior to this work, where we need $O(\timeSU)$ time for each problem. 
		The complexity $\timeSU$ is the time for solving \textsc{SuffixUpdate}, see \cref{tab:timeSU} for possible bounds, which can be in $O(\polyloglog n)$.
	}
	\label{tab:contribution}
	\begin{tabular}{lll}
		\toprule
		problem   & known & $O(\timeSU)$-time solution\\
		\midrule
		LRS array & $O(\lg^3 n)$~\cite{okanohara08online} & Section~\ref{sec:lrs} \\
		LZ77      & $O(\lg^3 n)$~\cite{okanohara08online} & Section~\ref{sec:lz77} \\
		MUS       & $O(\lg \sigma)$ amortized~\cite{mieno20minimal}          & Section~\ref{sec:mus} \\
		reverse LZ & $O(\lg^2 \sigma)$ amortized~\cite{sugimoto13reversedLZ} & Section~\ref{sec:reversedlz} \\
		overlapping reverse LZ & $O(\lg^2 \sigma)$ amortized~\cite{sugimoto13reversedLZ} & Section~\ref{sec:reversedlz} \\
		\bottomrule
	\end{tabular}
\end{table}

\section{Preliminaries}
\label{sec:prelim}
We assume the \emph{word RAM} with word size $w \geq \log n$ bits, where $n$ is the length of the input string.
Let $\Sigma$ denote an integer alphabet of size $\sigma = |\Sigma| = n^{O(1)}$.
An element of $\Sigma^*$ is called a \emph{string}.
Given a string $S \in \Sigma^*$, we denote its length with $|S|$,
its $i$-th letter with $S[i]$ for $i \in [1..|S|]$.
Further, we write $S[i..j] = S[i]\cdots S[j]$.
We write $\pali{S}$ to denote the \emph{inverted} string of $S$, i.e., $\pali{S} = S[|S|] S[|S|-1] \cdots S[1]$.
The \emph{suffix tree} $\STree(S)$ of a string $S$ is a compacted trie representing all suffixes of $S$.

In what follows, we fix a string $T[1..n]$ over $\Sigma$ as the input string to be processed online.
Here, online means that the algorithm reads $T$ letter by letter and maintains the result for the processed portion of the string.
All presented algorithms in this paper use $O(n)$ words of space, unless otherwise stated. Note that this means that the space is proportional to the size of the currently processed string (as opposed to the size of the entire string). 

\subsection{Breslauer--Italiano modification of Weiner's algorithm}\label{sec:weiner}

Weiner's algorithm computes the suffix tree $\STree(T)$ of a string $T[1..n]$ by processing it in backward direction and inserting suffixes from the shortest $T[n..]$ to the longest~$T[1..]$. 
The algorithm
augments suffix tree nodes with additional pointers called \textit{Weiner links} (W-links for short): 
a \emph{W-link} labeled with letter \(c\) from a node \(\alpha\) (denoted by \(\Weiner_c(\alpha)\)) points to the locus of string \(c \cdot \strlabel(\alpha)\), where \(\strlabel(\alpha)\) is the string spelled out by the path from the root to \(\alpha\). 
\(\Weiner_c(\alpha)\) is defined only if \(c\cdot \strlabel(\alpha)\) is a substring of the current string. 
If the locus of \(c\cdot \strlabel(\alpha)\) is a node, the W-link is called \emph{hard}, otherwise it is called  \emph{soft}. 
In the Breslauer--Italiano modification of Weiner's algorithm, a soft W-link is represented by a pointer to the closest descendant node of the locus of \(c \cdot \strlabel(\alpha)\). 

A round of Weiner's algorithm corresponds to the processing of a newly read letter.
At round $i$, the algorithm inserts the new suffix to the suffix tree, solving the following problem.

\problembox{\textsc{SuffixUpdate}\\
  \textbf{Input:} suffix tree $\STree(T[i+1..])$ and letter~$T[i]$.\\
  \textbf{Output:} $\STree(T[i..])$.
}

To this end, it locates the \textit{insertion point} where a new leaf should be attached. 
The insertion point is the locus of the longest prefix of $T[i..]$ already present as substring in $T[i+1..]$, it can be either an existing node or a locus on an existing edge. 
Computing the insertion point is the key operation of Weiner's algorithm, which we formulate as follows.

\problembox{\textsc{InsertionPoint}\\
  \textbf{Input:} suffix tree $\STree(T[i+1..])$ and letter~$T[i]$.\\
  \textbf{Output:} the locus of the longest prefix of $T[i..]$ appearing in $T[i+1..]$.
}

\begin{figure}[t]
  \centering
    \scalebox{.9}{
  \includegraphics[width=0.4\textwidth,page=1]{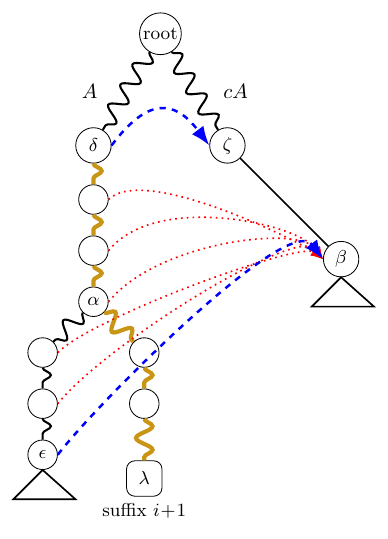}
  \hspace{2em}
  \includegraphics[width=0.4\textwidth,page=2]{img/suffixupdate.pdf}
}
  \caption{One round of Weiner's suffix tree construction algorithm: updating $\STree(T[i+1..])$ (left) to $\STree(T[i..])$ (right) by inserting the new suffix $T[i..]$.
    Dashed blue arrows represent hard W-links, dotted red arrows represent soft W-links, both for the letter~$c=T[i]$.
    The W-links on the path from $\alpha$ to $\epsilon$ in the right tree are not shown for the sake of clarity.
    Curly edges represent paths that may contain multiple nodes.
    Node $\alpha$ is  the closest ancestors of $\lambda$ having a W-link by $c$. This link can be soft (as in the figure) or hard (in case $\alpha=\delta$). 
If this link is soft, the algorithm creates a new node $\gamma$ which is an insertion point. 
If this link is hard, the insertion point is  $W_c(\alpha)$. 
    After creating $\gamma$, the W-links on the golden thick curly path from $\lambda$ to $\delta$ need to be updated (Steps~\ref{update1}--\ref{update2}).
  }
  \label{fig:weiner_suffix_update}
\end{figure}

Solving \textsc{InsertionPoint} or the whole update process is sometimes called the \textit{suffix tree oracle}~\cite{fischer15wexponential,takagi20fully,cole15trays}.

We summarize the steps implementing \textsc{SuffixUpdate} by Breslauer and Italiano's algorithm~\cite{breslauer13near}. 
The steps are also illustrated in Figure~\ref{fig:weiner_suffix_update}.
Let $c = T[i]$.

\begin{enumerate}
  \item Starting from the leaf $\lambda$ corresponding to suffix $T[i+1..]$,  find its lowest 
    ancestor $\alpha$ such that $\Weiner_{c}(\alpha)$ is defined. \label{it:marked_ancestor}
  \item Let $\beta=\Weiner_{c}(\alpha)$.  If $\Weiner_{c}(\alpha)$ is hard, then $\beta$ is the insertion point. 
    Otherwise, split the parent edge of $\beta$ and create a new node $\gamma$.  
    Copy all W-links from node $\beta$ to $\gamma$ which all become soft. \label{it:copyweiner}
  \item Create a new leaf $\lambda'$ for suffix $T[i..]$ as a child of the insertion point. 
\item Create a hard W-link $\Weiner_{c}(\lambda)=\lambda'$ and a soft W-link 
  $\Weiner_{c}(q)=\lambda'$ for each node $q$ on the path between $\lambda$ (excluded) up to $\alpha$ (excluded).  \label{update1}
  \item Create a hard W-link $\Weiner_{c}(\alpha)=\gamma$ and a soft W-link 
  $\Weiner_{c}(q)=\gamma$ for each node $q$ on the path between $\alpha$ (excluded) up to the first node $\delta$ for which $\Weiner_{c}(\delta)\neq \beta$ (excluded). \label{update2}
\end{enumerate}

Implemented naively, Step~\ref{it:marked_ancestor} can visit up to $O(n)$ nodes, Step~\ref{it:copyweiner} may have to copy $O(\sigma)$ W-links, and Steps~\ref{update1}--\ref{update2} can update the W-links of $O(n)$ nodes in the worst case.
However, assuming a constant-size alphabet, the total time complexity over all $n$ rounds can be shown to be $O(n)$ as follows.
Step~\ref{it:marked_ancestor} traverses $(\fnDepth(\lambda_{i+1}) - \fnDepth(\delta))$ nodes, where $\lambda_{i}$ stands for the leaf labeling the entire string $T[i..]$. 
By \cref{lem:weiner_depth} in Appendix~\ref{app:weiner}, we have $\fnDepth(\zeta)\leq\fnDepth(\delta)+1$, and then $\fnDepth(\lambda_{i+1}) - \fnDepth(\delta)\leq \fnDepth(\lambda_{i+1}) - \fnDepth(\zeta)+1 = \fnDepth(\lambda_{i+1}) - \fnDepth(\lambda_{i})+3$. Summing up over all rounds, Step~\ref{it:marked_ancestor} visits $O(n)$ nodes altogether. 
By a similar argument, Steps~\ref{update1}--\ref{update2} update up to $O(n)$ W-links as well.
Finally, Step~\ref{it:copyweiner} needs $O(n)$ time in total since there are at most $n$ hard W-links in the entire suffix tree and there are $O(1)$ W-links to copy at each round.

To achieve a nontrivial worst-case time bound per letter,
we need to efficiently implement Steps~\ref{it:marked_ancestor} and Steps~\ref{update1}--\ref{update2}.
Breslauer and Italiano \cite{breslauer13near} proposed to implement Step~\ref{it:marked_ancestor} by maintaining, for each alphabet letter $a$, the (dynamic) Euler tour list of the current suffix tree where nodes are \textit{marked} in color~$a$ according to $\Weiner_{a}$-links defined in them.
 Then, the ancestor node $\alpha$ is retrieved using two data structures: one supporting queries for the \textit{previous marked element} in a dynamic list,  applied to the Euler tour, and the other supporting \textit{dynamic lowest common ancestor} applied to the suffix tree itself.
 The first is solved using the data structure of Dietz and Raman~\cite{dietz91persistence} and the second using the data structure of Cole and Hariharan~\cite{cole05dynamic}.
 Altogether, this results in $O(\log\log n)$ time.
 We note in passing that for a constant-size alphabet, dynamic lowest common ancestor queries can be easily avoided, therefore only colored predecessor queries are critical.

To efficiently implement updates of Steps~\ref{update1}--\ref{update2}, the authors propose to distribute this work over multiple rounds of the algorithm (i.e. multiple letters of the input string) by updating only a constant number at each round. 
(For this reason, the algorithm is called \textit{quasi-real time}.) However, the nodes are updated in the order of increasing depth which guarantees no interference with the ongoing process of tree updating. 
As a result, Breslauer and Italiano modification spends $O(\log\log n)$ worst-case time on each letter. 

The central observation for this paper is that this modification, as well as the improvements in the following subsection,
maintain the suffix tree within this worst case-time per letter while keeping nearly complete functionality.
\begin{observation}
  The suffix tree updated by the Breslauer--Italiano algorithm
  is fully functional except that some W-links may be missing temporarily (due to the deferred updates in Steps~\ref{update1}--\ref{update2}).
\end{observation}

\subsection{Improvements of Breslauer--Italiano algorithm}\label{sec:improvements}

Altogether, the Breslauer--Italiano algorithm achieves worst-case $O(\sigma \log\log n)$ time per read letter, where $\sigma$ is the alphabet size. For $\sigma=O(1)$ this results in $O(\log\log n)$ time per read letter. 
This technique was the first to achieve a double logarithmic worst-case bound, as only an $O(\log n)$ bound had been known previously~\cite{amir05towards}. 

An obvious weakness of the Breslauer--Italiano approach is the linear dependency on the alphabet size. Subsequent works proposed improvements on this point. 
Kopelowitz~\cite{kopelowitz12online} proposed a randomized algorithm for maintaining the suffix tree online in $O(\log\log n + \log\log\sigma)$ worst-case \textit{expected} time per letter. 
Kucherov and Nekrich~\cite{kucherov17fullfledged} generalize the $O(\log\log n)$ bound of~\cite{breslauer13near} to log-sized alphabets, more precisely to $\sigma=O(\log^{1/4} n)$, using the \textit{generalized van Emde Boas data structure}  of Giora and Kaplan~\cite{10.1145/1541885.1541889}. An interesting improvement proposed in~\cite{kucherov17fullfledged} is that weak W-links are not maintained at all but are computed on the fly in the lazy fashion by using another combination of colored predecessor and LCA queries. This greatly simplifies the steps of the Breslauer--Italiano algorithm: Steps~\ref{it:copyweiner}, \ref{update1} and~\ref{update2} become constant time and only Step~\ref{it:marked_ancestor} remains to be implemented. The latter is essentially reduced to answering the colored predecessor queries on a dynamic colored list, for which the authors borrow an $O(\log\log n)$ solution of~\cite{10.1145/1541885.1541889} for up to $O(\log^{1/4} n)$ colors.  Mortensen~\cite{mortensen06fully} proposes a $O(\log^2\log n)$ solution for the dynamic colored predecessor problem for any number of colors, which can be plugged in as an alternative to compute \textsc{InsertionPoint}.

On the other hand, Fischer and Gawrychowski~\cite{fischer15wexponential} obtain $O(\log\log n + \frac{\log^2\log \sigma}{\log\log\log \sigma})$ time using a sophisticated technique of \textit{Wexponential search trees}. They store only some of the weak W-links and show how to implement Steps~\ref{it:marked_ancestor} and  \ref{it:copyweiner} within this time bound.  Steps~\ref{update1} and~\ref{update2} can be deamortized as in the Breslauer--Italiano algorithm. A summary of the known time complexities $\timeSU$ for processing \textsc{SuffixUpdate} is given in Table~\ref{tab:timeSU}.

\begin{table}[t]
\begin{minipage}{0.3\linewidth}
	\centering
	\caption{Time complexities $\timeSU$ for \textsc{SuffixUpdate}}
	\label{tab:timeSU}
\end{minipage}
\hfill
\begin{minipage}{0.55\linewidth}
	\newcommand*{\ditto}{--\raisebox{-0.5ex}{"}--}
	\begin{tabular}{l l | }
		\toprule
		Time complexity $\timeSU$ & Reference  and comments \\
		\midrule
		$O(\lg n)$ & \cite{amir05towards}  \\
		$O(\sigma \log\log n)$ & \cite{breslauer13near} \\
		$O(\log\log n + \log\log \sigma)$ & \cite{kopelowitz12online}, expected time \\
		$O(\log\log n + \frac{\log^2\log \sigma}{\log\log\log \sigma})$ & \cite{fischer15wexponential} \\
		$O(\log\log n)$ & \cite{kucherov17fullfledged}, for $\sigma = O(\log^{1/4} n)$  \\
\bottomrule
	\end{tabular}
\end{minipage}
\end{table}

Note finally that the idea of maintaining online an indexing data structure for the inverted string has been applied to indexes based on the Burrows--Wheeler transform (BWT)~\cite{burrows94bwt} as well. 
Similar to the suffix tree, it is more convenient to maintain a BWT-index for the inverted string, by adding only one suffix at each round~\cite{policriti18lz77,ohno17rlbwt}.
Updating the BWT-index boils down to maintaining a dynamic string that supports access, rank, and select queries, for which logarithmic solutions exist~\cite{navarro14dynamic}, which, however, are optimal~\cite{fredman89cell}.
Thus, there is no great hope for improving the time bounds for online BWT-index maintenance beyond logarithmic time per letter.
In what follows, we focus on a line of research that applies this online BWT-index construction for computing the LRS array, and show that we can obtain better time bounds by relying on the suffix tree construction instead.

\section{Longest Repeating Suffix Array}\label{sec:lrs}
From now on, we will be considering the online setting where the input string $T$ is provided letter by letter from left to right, and we will be maintaining a suffix tree of the \textit{inverted} string $\pali{T}$, using a variant of Weiner's algorithm. Thus, a \textit{suffix} of a current string will become an \textit{inverted prefix} of the actually indexed string, etc. In our presentation, we will consider both the string and its inverted image and alternate between the two depending on the context which should be clear to the reader. 

The \emph{longest repeating suffix array} (LRS) of a string $T[1..n]$ is an array $\LRS[1..n]$ such that for each position $i \in [1..n]$, $\LRS[i]$ is the length of the longest suffix of $T[1..i]$ that occurs at least twice in $T[1..i]$. In other words, $\LRS[i]$ is the length of the longest suffix of $T[1..i]$ occurring also in $T[1..i-1]$.

Okanohara and Sadakane~\cite{okanohara08online} were the first to study computing this array online, 
for which they obtained \(O(\log^3 n)\) worst-case time per letter. Prezza and Rosone~\cite{prezza20faster} obtained \(O(\log^2 n)\) \textit{amortized} time, with $O(n \lg n)$ worst-case delay per letter. 
Importantly, both approaches~\cite{okanohara08online} and~\cite{prezza20faster}  solve the problem under additional constraints on the space used by the algorithm: \textit{compact space} for the former and  zero-entropy \textit{compressed space} for the latter. 

We here show that, without the compact space constraint,  we can efficiently compute $\LRS$ in near-real-time as a by-product of maintaining $\STree$.
Observe that under Weiner's approach described in Section~\ref{sec:weiner}, the problem becomes the following. 

\problembox{\textsc{LongestRepeatingPrefix}\\
  \textbf{Input:} suffix tree $\STree(T[i+1..])$ and letter~$T[i]$.\\
  \textbf{Output:} the longest repeating prefix of $T[i..]$.
}

We solve \textsc{LongestRepeatingPrefix} by revisiting the methodology of Amir et al.~\cite{amir02online}.
They reduce the problem to \textsc{InsertionPoint}, since the string label of the insertion point is the longest prefix of $T[i..]$ that already appears in $T[i+1..]$.
Thus, updating $\STree(\pali{T})$ using Weiner's algorithm allows us to compute $\LRS$. 
When processing letter $T[i]$, we query \textsc{LongestRepeatingPrefix} on $\STree(\pali{T[1..i-1]})$ with letter $T[i]$. The answer is $\LRS[i]$.

\begin{theorem}
  The longest repeating suffix array $\LRS$ of a string $T[1..n]$ can be computed online in $O(\timeSU)$ worst-case time per letter.
  \label{thm:LRS-online}
\end{theorem}

\section{Longest Previous Factor Array}\label{sec:lpf}

The \emph{longest previous factor array} (LPF) of a string $T[1..n]$ is an array $\LPF[1..n]$ such that for each position $i \in [1..n]$, 
$\LPF[i]$ is the length of the longest prefix of $T[i..n]$ that has another occurrence starting at a position $j<i$.

We first describe how to convert $\LRS$ to $\LPF$ and vice versa in linear time \textit{offline}. 
A similar conversion between border and prefix arrays has been studied in the literature~\cite{DBLP:conf/iwoca/BlandKS13}. 
For that, we observe that we can restate the definitions of $\LRS$ and $\LPF$ as follows:
Let $X_i = \{ j \le i \mid j + \LPF[j] - 1 \ge i \}$ be the set of positions $j$ such that the longest previous factor starting at $j$ covers position $i$.
If $X_i \neq \emptyset$, $|X_i| = i - \min(X_i) + 1$.
Consequently, 
\[
\LRS[i] = 
\begin{cases}
|X_i| & \text{if } X_i \neq \emptyset, \\
0 & \text{otherwise}.
\end{cases}
\]
 
The problem of computing $\LRS$ from $\LPF$ is then reduced to finding the minimum in $X_i$ for each $i$.
For that, we focus on \emph{irreducible} LPF values, 
i.e., positions $i$ such that $\LPF[i] > \LPF[i-1]-1$. The important invariant is that irreducible LPF intervals of the form $[i..i+\LPF[i]-1]$ cannot nest but overlap or are disjoint. 
Thus, we can determine $\LRS[i]$ by maintaining the leftmost irreducible LPF interval covering position $i$.
To this end, we scan the irreducible LPF values from left to right. 
More precisely,
with a sweep-line $i \in [1..n]$ from left to right we maintain the leftmost irreducible LPF interval $[j..j+\LPF[j]-1]$ such that $i \in [j..j+\LPF[j]-1]$. 
Then $j$ is the minimum $j$ in $X_i$.

We can also compute $\LPF$ from $\LRS$ in linear time by scanning the $\LRS$ array left-to-right and computing $\LPF$ using the following rule. For each $i$ such that $\LRS[i] \le \LRS[i-1]$, 
we set $\LPF[j]=i-j$ for all $j\in [i-\LRS[i-1]..i-\LRS[i]]$.
For $i=n$, we complete the computation by setting $\LPF[j]=n-j+1$ for all $j\in [n-\LRS[n]+1..n]$.

\begin{theorem}
Offline conversions between the longest repeating suffix array $\LRS[1..n]$ and the longest previous factor array $\LPF[1..n]$ of a string $T[1..n]$ can be done in $O(n)$ time.
\end{theorem}

Like the LCP array~\cite{sadakane07compressed}, the LPF array can be recovered from the list of its $r$ irreducible values in linear time~\cite{bannai17square},
where it has been shown how to store the irreducible values in $2n + o(n)$ bits of space and still support constant-time access to either the $i$-th irreducible value or $\LPF[i]$, 
for any given $i$.
By an analogous argument, we can store the irreducible values of $\LRS$ within the same space bound with the same query functionality~\cite{prezza20faster}.
It is therefore natural to ask whether the above conversions can be done on the irreducible values only.
Actually, we can observe that the above algorithm for converting between $\LRS$ and $\LPF$ can be adapted to work on the irreducible values only, by skipping the assignments for reducible values.

\begin{corollary}
Offline conversions between the $O(r)$ irreducible LRS values and the $O(r)$ irreducible LPF values can be done in $O(r)$ time.
\end{corollary}

The left-to-right conversion of $\LRS$ into $\LPF$ can also be used to compute the $\LPF$ values online by computing $\LRS$ values online as described in Section~\ref{sec:lrs}, with a delay.
By delay we mean that given $T[1..i]$ is the currently processed text, then we have not yet computed the $\LPF[j]$ last values for the positions $j \in [i-\LRS[i]+1..i]$.
To turn our offline left-to-right conversion into (near-)real-time, we need to deamortize the sequence of assignments triggered by $\LRS[i]<\LRS[i-1]+1$. This can be done by distributing these assignments over multiple rounds and setting a constant number of values at each round, similarly to Steps~\ref{update1} and~\ref{update2} of Breslauer--Italiano algorithm (\cref{sec:weiner}). 
The delayed assignments can be stored in a FIFO data structure. 
By induction we observe that the maximal number of delayed assignments is bounded by $\max_i \LRS[i]$. 
We thus obtain the following result.

\begin{theorem}
  The longest previous factor array $\LPF$ of a string $T[1..n]$ can be computed online in $O(\timeSU)$ worst-case time per letter, with at most $\max_i \LRS[i]$ delayed ending values at each round.
\label{thm:LPF-online}
\end{theorem}

\section{Lempel--Ziv factorization (LZ77)}\label{sec:lz77}
The LZ77 factorization of a string \(T\) is defined as follows:
a factorization \(T = F_1 \cdots F_z\)
it is the \emph{LZ77 factorization} of \(T\) if each next factor \(F_x\), for \(x \in [1..z]\),
is either the first occurrence of a letter 
or the longest prefix of $F_x \cdots F_z$ that occurs at least twice in $F_1 \cdots F_x$.
The factorization can be written like a macro scheme~\cite{storer82lzss}, 
i.e., by a list storing either plain letters or pairs of referred positions and lengths, 
where a referred position is a previous text position from where the letters of the respective factor can be copied.

While practical compressors based the LZ77~\cite{ziv77lz} factorization such as gzip and 7zip work with a sliding window over the input text,
the original LZ77 factorization considers the entire text as the search buffer.
For the online setting,
Gusfield's textbook~\cite[APL~16]{gusfield97algorithms} presents an algorithm that runs in \(O(\log \sigma)\) amortized time per letter.
The main idea is to maintain the suffix tree of the text read so far using Ukkonen's algorithm~\cite{ukkonen95online}, and simultaneously compute each next factor by a top-down traversal of the tree. 
Since we can easily extract the LZ77 factorization from the $\LRS{}$ or $\LPF{}$ arrays,
algorithms of Sections~\ref{sec:lrs},\ref{sec:lpf} can be also used to compute it. 
Direct online algorithms for computing the LZ77 factorization have also been proposed, additionally focusing on space saving. 
In this line of research,
Starikovskaya~\cite{starikovskaya12computing} proposed an algorithm in \(O(\log^2 n)\) amortized time per letter
by returning to Ukknonen's suffix tree construction algorithm.
Subsequently, Yamamoto et al.~\cite{yamamoto14faster} obtained \(O(\log n)\) amortized time per letter by
constructing the directed acyclic word graph~\cite{blumer85dawg}, and finally
Policriti and Prezza~\cite{policriti15fast} achieved the same time bounds with a dynamic FM-index in zero-order entropy compressed space.

Here we diverge from this line of research by allowing $O(n)$ words of space but focusing on the worst-case time per letter.
Using the near-real-time computation of $\LRS$ (Section~\ref{sec:lrs}), 
it is straightforward to compute the LZ77 factorization within the same time bounds. 

Online computation of LZ77 factorization amounts to reporting if the current position $i$ extends the current factor or starts a new one. In the latter case, the previous factor ends at position $i-1$ and $i$ is the first position of a new factor. 
We can determine this by comparing $i-\LRS[i]$ with the beginning $s_x$ of the current factor $F_x$. If $i-\LRS[i]+1\leq s_x$, then $i$ extends $F_x$, otherwise $F_x$ ends at $i-1$ and $i$ starts a new factor $F_{x+1}$. 
We observe that $\LRS[i]=0$ if and only if $T[i]$ is the leftmost occurrence of a letter.

\begin{theorem}
  The LZ77 factorization of a string $T[1..n]$ can be  computed online in $O(\timeSU)$ worst-case time per letter.
\end{theorem}

\section{Minimal Unique Substrings}\label{sec:mus}
Given a substring $S$ of a text $T[1..n]$, let $\numocc_T(S)$ denote the number of occurrences of $S$ in $T$.
A substring $S$ of $T$ is called \emph{unique} in $T$ if $\numocc_T(S) = 1$ and  \emph{repeating} in $T$ if $\numocc_T(S) \ge 2$. 
A unique substring $S$ of $T$ is called \emph{a minimal unique substring}
of $T$ if $S$ is unique and any proper substring of $S$ is repeating in $T$.
Since a unique substring $S$ of $T$ has exactly one occurrence in $T$,
it can be identified with a unique interval $[\ell..r] \subseteq [1..n]$ with $S = T[\ell..r]$.
We denote the set of intervals corresponding to the MUSs of $T$
by $\MUS(T) = \{[s..t] \mid T[s..t]~\text{is a MUS of}~T\}$.
By the definition of MUSs, 
$[s..t] \in \MUS(T)$ if and only if
(a) $T[s..t]$ is unique in $T$,
(b) $T[s+1..t]$ is repeating in $T$, and
(c) $T[s..t-1]$ is repeating in $T$.

Mieno et al. \cite{mieno20minimal} present an algorithm for computing MUSs in a sliding window and, in particular, study how $\MUS(T)$ can be updated when a new letter $T[j+1]$ is appended to $T[1..j]$. We summarize these updates here, slightly modifying the description of  \cite{mieno20minimal}. 
Recall that $\LSuf{j}$ denotes the longest repeating suffix of $T[1..j]$, and let $\SSuf{j}$ be the shortest suffix of $T[1..j]$ with exactly two occurrences in $T[1..j]$ (that is, having exactly one previous copy).  
Since $\LSuf{j}$ can admit the empty string, $\LSuf{j}$ always exists. 
On the other hand, $\SSuf{j}$ may not exist when no suffix of $T[1..j]$ occurs exactly twice, i.e., $\numocc_{T[1..j]}(T[i..j]) \neq 2$ for all $i \in [1..j]$.

The size of $\LSuf{j}$ and the ending position of the previous copy of $\SSuf{j}$ (if $\SSuf{j}$ exists) are the two parameters we need to know in order to specify all possible modifications to $\MUS(T[1..j])$ when $T[j+1]$ is appended to $T[1..j]$. 
These modifications can consist of one deletion of a MUS and up to three additions of new MUSs, specified below. 

\begin{enumerate}[label=(\arabic*):,ref=(\arabic*),series=musCond]  
  \item 
  If $|\LSuf{j+1}| \le |\LSuf{j}|$, then $[j+1-|\LSuf{j+1}| ..j+1]$ is a new MUS to be added to $\MUS(T[1..j+1])$.

  \myProof{} By definition of $\LSuf{j+1}$, on the one hand, $T[j+1-|\LSuf{j+1}| ..j+1]$ is unique. 
  On the other hand, both $T[j+1-|\LSuf{j+1}| ..j]$ and $T[j+1-|\LSuf{j+1}|+1 ..j+1]$ are repeating in $T[1..j+1]$. The former is repeating as it must be a suffix of $T[j-|\LSuf{j}|+1..j]$. 

    \label{case:mus_lemma_four}
\item If $\SSuf{j+1}$ exists, let $T[s..q]$ be its previous copy ($q< j+1$). Then $[s..q]$ is in $\MUS(T[1..j])$ but not in $\MUS(T[1..j+1])$ and therefore should be deleted. 
  
    \myProof{} $T[s..q]$ was unique in $T[1..j]$ as it occurs exactly twice in $T[1..j+1]$. $T[s+1..q]$ must be repeating in $T[1..j]$ as $T[s..q]=\SSuf{j+1}$ is the \textit{shortest} suffix of $T[1..j+1]$ occurring twice, therefore $T[s+1..q]$ must occur at least three times in $T[1..j+1]$ and therefore at least twice in $T[1..j]$. On the other hand, $T[s..q-1]$ is a repeating suffix of $T[1..j]$. Therefore, $T[s..q]$ was a MUS in $T[1..j]$ that is no longer a MUS in $T[1..j+1]$.

  Note that $T[s..q]$ may overlap $T[j+1-|\SSuf{j+1}|+1..j+1]$, in particular, it may happen that $q=j$. 
  In this case, $T[s..j]=\SSuf{j+1}=T[s+1..j+1]$ and therefore $T[s..j+1]=T[j+1]^k$, $k=|\SSuf{j+1}|+1$, is the single letter run which is unique in $T[1..j+1]$. 
  Observe that in this case, we have $\SSuf{j+1}=\LSuf{j+1}$.
  The situation is illustrated on the right of Figure~\ref{fig:muslem6}.

If $\SSuf{j+1}$ exists, one or two new MUSs can appear in this case, specified in the following and illustrated on the left in Figure~\ref{fig:muslem6}. 
\label{case:mus_delete}
\item A new potential MUS to be added to $\MUS(T[1..j+1])$ becomes $[s..q+1]$ provided that no (shorter) MUS in $T[1..j]$ ends at $q+1$. 
  
  \myProof{} Since $T[s..q]=\SSuf{j+1}$ has a single non-suffix occurrence in $T[1..j+1]$, $T[s..q+1]$ is unique in $T[1..j+1]$. Since $T[s..q]$ is repeating in $T[1..j+1]$, the shortest unique suffix $T[\ell .. q+1]$ of $T[s..q+1]$ is a MUS in $T[1..j+1]$. If $\ell=s$, this MUS is a new one, otherwise it already 
  existed in $T[1..j]$. 
  
\label{case:mus_extension}
  \item A new potential MUS to be added to $\MUS(T[1..j+1])$ becomes $[\ell - 1..q]$ where $T[\ell..q]=\LSuf{j+1}$, provided that $\ell\geq 1$ and no (shorter) MUS in $T[1..j]$ starts at $\ell-1$.

    \myProof{} By definition of $\LSuf{j+1}$, $T[\ell - 1..q]$ is unique in $T[1..j+1]$ and $T[\ell ..q]$ is repeating. Therefore, the shortest unique prefix $T[\ell - 1..p]$ of $T[\ell - 1..q]$ is a MUS in $T[1..j+1]$. If $p=q$, this MUS is a new one, otherwise it already 
    existed in $T[1..j]$. 
    \label{case:mus_ext_left}
\end{enumerate}

\begin{figure}
  \centering
  \includegraphics[width=0.64\linewidth,valign=t]{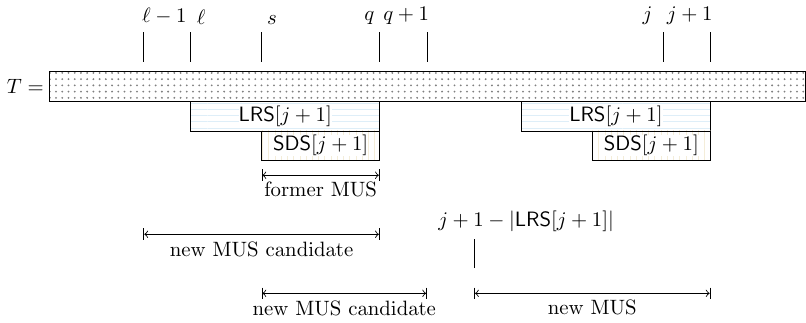}
  \includegraphics[width=0.34\linewidth,valign=t]{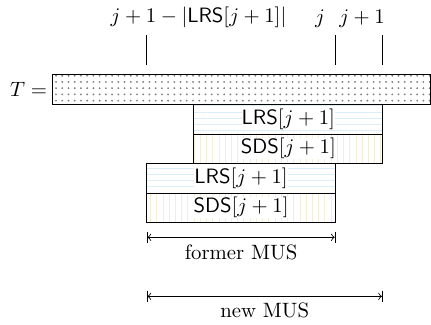}
	\caption{Illustration 
of Cases~\ref{case:mus_lemma_four} -- \ref{case:mus_ext_left}. Left: new MUS (Case~\ref{case:mus_lemma_four}), former MUS (Case~\ref{case:mus_delete}) and two potential new MUSs (Cases~\ref{case:mus_extension} and~\ref{case:mus_ext_left}),
		Right: particular case of Case~\ref{case:mus_delete} when $q=j$.
}
  \label{fig:muslem6}
\end{figure}

We now turn to the implementation of the above updates. First, we need to represent the current set of MUSs. Knowing that MUSs cannot nest, one can bijectively map the starting position of a MUS to its ending position, and vice versa. 
The authors of \cite{mieno20minimal} used this insight to create two arrays to maintain MUSs via the bijection and its reverse.
This allows them to retrieve a MUS by querying one of its end points, and adding/removing a MUS, all in constant time per operation. 

To implement the updates online, the authors of \cite{mieno20minimal} maintain $\STree(T[1..j])$ with Ukkonen's suffix tree construction algorithm~\cite{ukkonen95online}. While reading a new letter $T[j+1]$ from the text, they update $\STree(T[1..j])$ and compute the locus of $\LSuf{j+1}$ and $\SSuf{j+1}$, 
which they call \emph{active points}, and their string lengths  $|\LSuf{j+1}|$ and $|\SSuf{j+1}|$ respectively.  Computing $\LSuf{j+1}$ is a direct by-product of Ukkonen's algorithm, whereas computing $\SSuf{j+1}$ requires some additional work. Using this information, the set of MUSs is updated using the above description.

We now show how to efficiently implement the updates when we maintain the suffix tree online with 
Weiner's algorithm applied to the inverted text, as discussed in Section~\ref{sec:prelim}. 
Maintaining $\LSuf{i}$ has been studied in Section~\ref{sec:lrs}. To maintain $\SSuf{i}$, we need to solve the following problem. 

\problembox{\textsc{ShortestDoubleSuffix}\\
  \textbf{Input:} suffix tree $\STree(\pali{T[..j]})$ and letter~$T[j+1]$.\\
  \textbf{Output:} the shortest prefix of $\pali{T[..j+1]}$ that has exactly one occurrence in $\pali{T[1..j]}$, if it exists.
}

Assume a new letter $T[j+1]$ is appended to $T[1..j]$. 
As explained in Section~\ref{sec:lrs}, the locus of $\LSuf{j+1}$ corresponds to the insertion point of Weiner's algorithm (see Section~\ref{sec:weiner}). 
If the insertion point is an already existing node in $\STree(\pali{T[1..j]})$, then $\LSuf{j+1}$ occurs at least three times in $T[1..j+1]$ and $\SSuf{j+1}$ does not exist. 
Otherwise, the locus of $\LSuf{j+1}$ becomes a node in $\STree(\pali{T[1..j+1]})$ with exactly two leaves. 
Therefore, in this case, $\SSuf{j+1}$ exists and its locus is located on the parent edge to the locus of $\LSuf{j+1}$ and thus is obtained ``for free'' in constant time, simplifying the algorithm of \cite{mieno20minimal}.

Once $\LSuf{j+1}$ and $\SSuf{j+1}$ are found, we update the set of MUSs implementing the modifications described above:

\begin{enumerate}
	\item Add a new MUS $[j+1-|\LSuf{j+1}|..j+1]$.
	\item If $\SSuf{j+1}$ exists, find its non-suffix occurrence $[s..q]$: $s$ is the label of the leaf of the locus of $\SSuf{j+1}$ that already existed in $\STree(T[1..j])$, and $q=s+|\SSuf{j+1}|-1$. 
	Delete the MUS $[s..q]$. 
	\item If there is no MUS ending at position $q+1$, add the MUS $[s..q+1]$.
	\item If $q-|\LSuf{j+1}|\geq 1$ and there is no MUS starting at position $q-|\LSuf{j+1}|$, add the MUS $[q-|\LSuf{j+1}|..q]$. 
\end{enumerate}

Following \cite{mieno20minimal}, we maintain
two integer arrays $X_1$ and $X_2$ of size \(n\) mapping starting positions of the MUSs to their ending positions, and vice versa. In our online setting, arrays $X_1$ and $X_2$ are dynamic and grow by one entry to the right at each newly read letter. 
We implement $X_1$ and $X_2$
using the standard doubling technique (e.g. \cite[Section~3.2]{mehlhorn08algorithms}) which allocates a new array fragment of size \(n\) when the size of the current string reaches \(n\), thus doubling the occupied memory. 
While this implementation maintains the $O(n)$ memory, the waisted memory is also up to $O(n)$. To reduce the waisted memory, one can use 
the \emph{resizeable array} of Brodnik~\cite{brodnik99resizable} which supports the growable array operations in $O(1)$ worst-case time while having only up to $O(\sqrt{n})$ of waisted memory.
We summarize this section with the following final result.

\begin{theorem}
  We can maintain the set of minimal unique substrings $\MUS(T[1..n])$ of a string $T[1..n]$ online in $O(n)$ space and $O(\timeSU)$ worst-case time per letter.
\end{theorem}

\section{Reversed Lempel--Ziv}\label{sec:reversedlz}

The reversed {LZ} factorization was introduced by Kolpakov and Kucherov~\cite{kolpakov09gappedpalindromes}
as a helpful tool for detecting \emph{gapped palindromes}, 
i.e.,  substrings of a given text~$T$ of the form $\pali{S}GS$ for two strings $S$ and $G$.  
The reversed {LZ} factorization of $T$ is defined recursively as follows:
a factorization $T = F_1 \cdots F_z$ 
is the \emph{reversed {LZ} factorization} of $T$ if each factor $F_x$, for $x \in [1..z]$,
is either the first occurrence of a letter or the longest prefix of $F_x \cdots F_z$ that has an inverted copy $\pali{F_x}$ as a substring of $F_1 \cdots F_{x-1}$.
Like LZ77, the factorization is a macro scheme~\cite{storer82lzss}.
Among all variants of such a left-to-right parsing using the reversed substring as a reference to the formerly parsed part of the text, 
the greedy parsing achieves optimality with respect to the number of factors~\cite[Theorem~3.1]{crochemore14note}
since the reversed occurrence of~$F_x$ can be the prefix of any suffix in $F_1 \cdots F_{x-1}$, 
and thus fulfills the suffix-closed property~\cite[Definition~2.2]{crochemore14note}.

Kolpakov and Kucherov~\cite{kolpakov09gappedpalindromes} also gave an algorithm computing the reversed {LZ} factorization in $O(n \lg \sigma)$ time using $O(n \lg n)$ bits of space,
by applying Weiner's suffix tree construction algorithm~\cite{weiner73linear} on the reversed text~$\pali{T}$.
Later, Sugimoto et al. \cite{sugimoto13reversedLZ} presented an online factorization algorithm running in $O(n \lg^2 \sigma)$ time using $O(n \lg \sigma)$ bits of space.
Recently, Köppl~\cite{koppl21reversed} improved the time complexity to $O(n)$ while keeping the space usage of $O(n \lg \sigma)$ bits. 
While the last approach works offline, the former results are online algorithms with amortized time bounds per letter.

\subparagraph*{Reversed LZ factorization without self-references}
In the remainder of this section, we follow the ideas of the algorithm of~\cite{kolpakov09gappedpalindromes} to show how to compute the reversed {LZ} factorizations online in $O(\timeSU)$ worst-case time per processed letter.
As in the previous sections, we maintain the suffix tree of the reversed text~$\pali{T}$ online using Weiner's algorithm.
The new idea is that we require the capability for pattern matching coupled with tree construction.
In parallel to extending the suffix tree, we are matching a pattern which is the LZ factor we are currently parsing. 
The task is to find the locus of the longest prefix of the pattern that can be read by traversing the suffix tree top-down.
For that, we maintain the currently computed locus of a matching prefix of the pattern.
After reading a new letter of the text, we match this letter by making one traversal step of the suffix tree. 
If the letter cannot be matched, the parsing of the current factor terminates and the new one starts. 
Then the suffix tree is extended with the new letter.

The important point, however, is that we are not allowed to traverse nodes that have been created after we started the pattern matching process.
To ensure that, we borrow the idea of timestamping suffix tree nodes~\cite[Section~3]{amir02online}.
Namely, we store a timestamp in each node of $\STree(\pali{T})$ indicating when it has been created,
which is identical to the largest suffix of $\pali{T}$ at the leaves of its subtree.
By doing so, we can abort the pattern matching process if we reach a node whose timestamp is larger than the ending position of the previous factor.
The timestamps can be maintained while constructing the suffix tree with Weiner's algorithm in constant additional time since the suffix length of a newly created leaf is known at the time of its creation, and affects only the timestamp of its parent node.

Recall that 
if we follow the Breslauer-Italiano approach to maintaining weak W-links, we deamortize the update of multiple weak W-links at an individual round by updating a constant number (at least two) of weak W-links at each round in the order of increasing depth. Importantly, this process does not interfere with the pattern matching process we are doing in parallel, as the latter only requires the timely creation of new nodes which in turn is governed by hard W-links alone updated in real time.

Finally, the pattern matching process incurs an additional cost which depends on how the suffix tree maintenance algorithm is implemented. In all implementations mentioned in Section~\ref{sec:prelim} (see Table~\ref{tab:timeSU}), the suffix tree traversal step is dominated by the update step. Note that the tree traversal step essentially reduces to accessing the corresponding child of a branching node, and a tree update step includes creation of a new leaf which, in most implementations, subsumes the parent-to-child access. 

We conclude with the following result.

\begin{theorem}\label{thm:reversedlz}
  We can compute the reversed {LZ} factorization of a string $T[1..n]$ online in $O(n)$ space and $O(\timeSU)$ worst-case time per letter.
\end{theorem}

\subparagraph*{Reversed LZ factorization with self-references}

Sugimoto et al. \cite{sugimoto13reversedLZ} presented a variant of the reversed {LZ} factorization
called the \textit{reversed {LZ} factorization with self-references}. In that version, a factor $F_x$ and its inverted copy $\pali{F_x}$ are allowed to overlap~\cite[Definition~4]{sugimoto13reversedLZ}.
In detail, 
a factorization $F_1\cdots F_z = T$ is the \emph{overlapping reversed {LZ} factorization} of $T$ 
if each next factor $F_x$ 
is the first occurrence of a letter 
or the longest prefix of $F_x \cdots F_z$
with the property that each non-empty prefix $F$ (including $F_x$ itself) of $F_x$ has an inverted copy $\pali{F}$ as a substring of $F_1 \cdots F_{x-1}F$.
Observe that a factor $F_x$ overlaps its inverted copy $\pali{F_x}$ if and only if $F_1 \cdots F_{x}$ has a palindromic suffix of length\footnote{Our definition slightly differs from that of \cite{sugimoto13reversedLZ} in that the inverted copy does not have to end \textit{before} the end of $F_x$, i.e., $F_x$ itself can be a palindrome.}
between $|F_x|$ and $2|F_x|-1$.

While the reversed {LZ} factorization with self-references cannot be directly used to encode the string, it can still be computed online using our technique together with Manacher's algorithm~\cite{manacher75new}. 
Manacher's algorithm is a linear-time algorithm to compute all maximal palindromes in a string, 
it reads the input string online and maintains as an invariant the longest palindromic suffix of the current string. By a standard deamortization argument, Manacher's algorithm can be made real time \cite{galil76real}, spending constant worst-case time on each new letter.

The idea is then to run Manacher's algorithm in parallel with  the algorithm of Theorem~\ref{thm:reversedlz}. 
Maintaining the longest palindromic suffix 
covers the case of self-overlapping factor whereas the algorithm of Theorem~\ref{thm:reversedlz} covers the case of non-overlapping reversed factor. We then keep at each step the longest factor between those obtained by the two algorithms. See Figure~\ref{fig:OVrevlz} for an illustration.

\begin{figure}[t]
  \begin{minipage}{0.4\linewidth}
  	\scalebox{.8}{
	\includegraphics{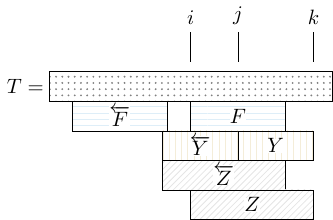}
}
  \end{minipage}
  \begin{minipage}{0.55\linewidth}
	\caption{Illustration of the reversed LZ factorization with self-references.
		We maintain both the longest non-overlapping reversed LZ factor $F$ starting at position $i$ and the longest palindromic suffix $\pali{Y} Y$ starting before $i$. The one which extends further right defines the factor. Here the suffix palindrome $\pali{Y} Y$ defines the new factor $Z$. }
	\label{fig:OVrevlz}
  \end{minipage}
\end{figure}

\begin{theorem}
	We can compute the reversed {LZ} factorization with self-references of a string $T[1..n]$ online in $O(n)$ space and $O(\timeSU)$ worst-case time per letter.
\end{theorem}

\section{Conclusion}
We studied applications of the Breslauer--Italiano deamoritization of Weiner's algorithm to obtain near-real-time online algorithms for several string processing problems.
The crucial observation was that the deamortization of the suffix tree construction process does not interfere with the functionality of the suffix tree needed to solve these problems.
We gave online algorithms computing the longest repeating suffixes, the Lempel--Ziv factorization, the set of minimal unique substrings, and the reversed Lempel--Ziv factorization (with and without self-references) in $O(n)$ space and $O(\timeSU)$ worst-case time per letter, where $\timeSU$ depends on the implementation of the suffix tree construction algorithm.

\bibliographystyle{plainurl}

\appendix
\section{Details of Weiner's suffix tree construction algorithm}\label{app:weiner}

Here we prove the following two statements that are needed to obtain the claims and the time complexity of Weiner's algorithm of  \cref{sec:weiner}. 
\begin{lemma}\label{lem:weiner_depth}
  The height of $\STree(T[i..n])$ is the height of $\STree(T[i+1..n])$ increased by at most one.
\end{lemma}
\begin{proof}
 At each round, Weiner's algorithm inserts a new suffix to the suffix tree.
 This insertion adds one leaf and at most one internal node~$\alpha$.
 Given we created this node $\alpha$, all nodes in the subtree of $\alpha$ have their depth increased by one,
 and the height of the suffix tree increases by at most one.
\end{proof}

\begin{lemma}
	Assume an internal node $\zeta$ in $\STree(T[i..])$ is the locus of substring $c X$, 
i.e., $\strlabel(\zeta) = c X$. 
  Then
  $X$ is the locus of a node $\delta$ in $\STree(T[i+1..])$, 
  in particular, $\Weiner_c(\delta) = \zeta$ in $\STree(T[i..])$ and the depth of $\zeta$ is at most the depth of $\delta$ plus one.
\end{lemma}
\begin{proof}
  If $\zeta$ is already an internal node in $\STree(T[i..])$, then the suffix link of $\zeta$ points to a node $\delta$ with string label $X$ in $\STree(T[i+1..])$.

  Now assume that $\zeta$ is not an internal node in $\STree(T[i+1..])$ which can only happen when $c = T[i]$.
Since $\zeta$ is an internal node  in $\STree(T[i..])$, there are distinct letters $d_1$ and $d_2$ such that $X d_1$ and $X d_2$ are substrings of $T[i+1..]$.
  Hence, $\delta$ must be an internal node in $\STree(T[i..])$.

  To show that the depth of $\zeta$ cannot be more than the depth of $\delta$ plus one, 
observe that if $\lambda$ is an ancestor of $\zeta$ with string label $c Y$, then there exists a node in $\STree(T[i+1..])$ with string label $Y$ that is an ancestor of $\delta$.
  Therefore, every node along the path to $\zeta$ maps by the suffix link to a distinct node on the path to $\delta$, except the possible node labeled $c$ which maps to the root. 
Consequently, the number of ancestors of $\zeta$ is bounded by the number of ancestors of $\delta$ plus one.
\end{proof}

\section{Time complexities for downwards traversal}\label{app:timeSUTraverse}

We here explicitly address the pattern matching problem, on which the solutions in \cref{sec:reversedlz} for the reversed Lempel--Ziv factorization rely.
In detail, we rely on answering two types of queries while maintaining the suffix tree online.

\problembox{\textsc{PatternMatch}\\
    \textbf{Input:} suffix tree $\STree(T[1..j-1])$, locus of $P \in \Sigma^*$ in it, letter $T[j]$, and a letter $c$.\\
    \textbf{Output:} $\STree(T[1..j])$ and the locus of $Pc$ in it.
}
This problem asks, for a given node~$v$ and letter~$c$, to retrieve $v$'s child edge whose first letter matches $c$.
Let $\timeTraverse$ denote the time complexity for solving \textsc{PatternMatch}.
The classic solution to this problem is to maintain explicit child pointers for each node in the suffix tree by a binary search tree, which allows answering \textsc{PatternMatch} in $O(\lg \sigma)$ time.
However, this time bound is subsumed by $\timeSU$ only for Weiner's algorithm presented in the first row.
Unfortunately, not all solutions presented in the table make explicit the time in which they can solve \textsc{PatternMatch}.
Instead of verifying each solution separately, we propose a general technique to obtain the desired time bounds for all known solutions.
The idea is to keep each suffix tree representation as it is.
However, we augment each node with a pointer to a copy of itself.
Its copy belongs to a trie which is maintained in parallel to the suffix tree, and takes the same topology as the suffix tree.
Vice versa, we store a pointer from the copy back to the original node such that we can always jump between suffix tree node and trie node in constant time.
Using this mapping, we can answer \textsc{PatternMatch} queries on the trie instead of the suffix tree.
Fortunately, it is possible to find, for each mentioned suffix tree representation, a trie implementation whose time bound for queries and updates $\timeTraverse$ is subsumed by the respective $\timeSU$.
We have summarized the results in Table~\ref{tab:timeSUTraverse}.

Consequently, we can solve \textsc{PatternMatch} in $O(\timeSU + \timeTraverse) = O(\timeSU)$ time per letter, and 
thus obtain the time bounds claimed in \cref{sec:reversedlz}.

\begin{table}[t]
 \centering
 \caption{Time complexities $\timeSU$ of \cref{tab:timeSU} and $\timeTraverse$ for solving \textsc{SuffixUpdate} and a downwards traversal. See \cref{app:timeSUTraverse} for details.}
 \label{tab:timeSUTraverse}
 \newcommand*{\ditto}{--\raisebox{-0.5ex}{"}--}
 \begin{tabular}{ll|ll|l}
   \toprule
   \multicolumn{2}{c}{suffix tree} & \multicolumn{2}{c}{trie} 
   \\\cmidrule(lr){1-2} \cmidrule(lr){3-4}
   $\timeSU$ & Ref.\ & $\timeTraverse$ & Ref.\ & caveat\\
   \midrule
   $O(\lg n)$ & \cite{amir05towards} & $O(\log \sigma)$ & binary search \\
   $O(\sigma \log\log n)$ & \cite{breslauer13near} & \multicolumn{1}{c}{\ditto{}} & \multicolumn{1}{c}{\ditto{}} \\
   $O(\log\log n + \log\log \sigma)$ & \cite{kopelowitz12online} & $O(\log \log \sigma)$ expected & y-fast trie~\cite{willard83yfast} & randomized\\
   $O(\log\log n + \frac{\log^2\log \sigma}{\log\log\log \sigma})$ & \cite{fischer15wexponential} & $O(\frac{\log^2\log \sigma}{\log\log\log \sigma})$ & Wexp trie~\cite{fischer15wexponential} \\
   $O(\log\log n)$ & \cite{kucherov17fullfledged} & $O(1)$ & $q$-heaps \cite{10.1145/1541885.1541889} & $\sigma = O(\log^{1/4} n)$\\
   \bottomrule
 \end{tabular}
\end{table}

\end{document}